\documentclass[conference,letterpaper]{IEEEtran}

\usepackage[utf8]{inputenc} 
\usepackage[T1]{fontenc}
\usepackage{url}
\usepackage{ifthen}
\usepackage{cite,enumitem,soul}
\usepackage[cmex10]{amsmath} 

\usepackage{xcolor}
\usepackage{graphicx}
\usepackage{amsthm}
\usepackage[colorlinks=true, allcolors=blue]{hyperref}
\usepackage{color}
\usepackage{mathtools,leftindex,tensor}
\usepackage[version=4]{mhchem}
\usepackage{tabularx}
\usepackage{caption}
\usepackage{multirow}
\usepackage{subcaption}
\usepackage{adjustbox}
\usepackage{bbm}

\usepackage[utf8]{inputenc}
\usepackage{fourier} 
\usepackage{array}
\usepackage{makecell}
\usepackage{float}
\usepackage{soul}
\usepackage{upgreek}
\usepackage[makeroom]{cancel}
\mathtoolsset{showonlyrefs}

\renewcommand{\arraystretch}{1.2}
\newtheorem{theorem}{Theorem}

\DeclareMathOperator*{\argmin}{arg\,min}

\newcommand{\expect}[1]{{\mathbb E} \left[ #1 \right]}

\begin{document}
\title{On Optimal Batch Size in Coded Computing} 


\author{%
  \IEEEauthorblockN{Swapnil Saha}
\IEEEauthorblockA{
                     Rutgers University\\
                    Email: swapnil.saha@rutgers.edu}
  \and
  \IEEEauthorblockN{Emina Soljanin}
  \IEEEauthorblockA{
                      Rutgers University\\
                     Email:emina.soljanin@rutgers.edu}
  \and
  \IEEEauthorblockN{Philip Whiting}
   \IEEEauthorblockA{
                     Macquarie University\\
                     Email:philipawhiting@gmail.com}
}


\maketitle


\begin{abstract}
We consider computing systems that partition jobs into tasks, add redundancy through coding, and assign the encoded tasks to different computing nodes for parallel execution. The expected execution time depends on the level of redundancy. The computing nodes execute large jobs in batches of tasks. We show that the expected execution time depends on the batch size as well. The optimal batch size that minimizes the execution time depends on the level of redundancy under a fixed number of parallel servers and other system parameters. Furthermore, we show how to (jointly) optimize the redundancy level and batch size to reduce the expected job completion time for two service-time distributions. The simulation presented helps us appreciate the claims.  
\end{abstract}

\section{Introduction}\label{Sec:Intro}

The development of distributed computing frameworks such as Apache Spark \cite{zaharia2010spark}, MapReduce \cite{dean2008mapreduce}, alongside commercial cloud-based platforms, e.g., Google Cloud Platform (GCP) and Amazon EC2, has led to a significant speedup in large-scale AI/ML computing. Such systems partition jobs into tasks and assign them to different computing nodes for parallel execution. However, this parallel execution makes the system unpredictable due to random fluctuation of nodes' slowdown and failures \cite{dean2013tail}. The slowest (straggling) workers determine the job completion time.

A common approach to speed up distributed computing is to add redundant tasks, most recently, through erasure coding \cite{lee2017speeding,kiani2018exploitation,ferdinand2018hierarchical,aktacs2019straggler,aktas2017effective,aktas2017simplex,peng2021diversity,ozfatura2019speeding,adikari2023straggler,wang2021batch,reisizadeh2019coded}. Coded systems divide jobs into $k \geq 1$ tasks and encode them into $n \geq k$ tasks so that execution of any $k$ tasks is sufficient for job completion. As usual, the code rate $R$ is defined as $R=k/n$.  The level of redundancy can vary from minimum redundancy (splitting the job equally among workers with $k=n$) to maximum redundancy (replicating the
job with workers with $k=1$). The optimal level of redundancy depends on the nature of the randomness of the workers and how that random nature changes with the task size \cite{peng2021diversity}. 

To optimize different metrics and due to various constraints, distributed systems process tasks in batches \cite{ferdinand2018hierarchical,kiani2018exploitation,adikari2023straggler,wang2021batch,wang2019scalable,behrouzi2020efficient,kar2021throughput,kar2020delicate,he2010comet}. In particular, \cite{wang2019scalable} used a task batching technique to reduce communication overhead for scalability, \cite{behrouzi2020efficient} used the variation in batching technique (overlapping and non-overlapping) to lower the expected job execution time, \cite{kar2021throughput,kar2020delicate} addressed the throughput of the distributed computing system based on batch processing. A significant amount of research has focused on how to use stragglers' partial execution to reduce job completion time based on the task batching approach \cite{ferdinand2018hierarchical,kiani2018exploitation,adikari2023straggler,wang2021batch}. For example, \cite{wang2021batch, kiani2018exploitation} proposed a batching technique in coded distributed matrix multiplication, \cite{ferdinand2018hierarchical} designed a coding scheme that partitioned the computation into layers (batches) and applied different erasure code rates to each layer. These works treat batch size as the system hyperparameter and choose the optimum one based on exhaustive search.

Here, we show that the level of redundancy affects even the choice of an optimal batch size. To appreciate the right choice of batch size, consider two extreme cases of redundancy for fixed $n$ parallel working servers: 1) splitting, where all $n$ workers need to finish the task to complete the job, and 2) replication, where any $1$ of $n$ workers needs to finish the task. We prove (a possibly intuitive result) that splitting performs best with as large a batch size as possible, whereas replication performs best with as small a batch size as possible. We continue our analysis to find the optimal batch size for a given level of redundancy and optimize the batch size and redundancy level simultaneously to reduce the expected job competition time. 

The remainder of the paper is organized as follows. Sec.\ref{sec:system_Model} presents the distributed system model. Sec.~\ref{sec:evaluation_metric} formulates the problem. 
In Sec \ref{sec:shifted_exponential}, we answer questions of interest discussed above for the Shifted Exponential service-time PDF. In the extended version of this work \cite{optimal_batch_size_coded_computing}, we present this analysis for the PDF of Bi-Modal service time and provide all omitted proofs. Tables \ref{tab:red_vs_batch} and Table \ref{tab:summ_table_b_R} summarize our results for both PDFs in service time. 

\section{Distributed System Model}\label{sec:system_Model}



\noindent
\ul{Job and Task Size:} We assume the job can be split into tasks up to some minimum task size. We refer to the minimum task size as \emph{computing unit (CU)}. We will express the size of the job and the task sizes in the number of CUs they contain. For example, consider matrix-vector multiplication $y=Ax$, where $y \in \mathcal{R}^3$ is the vector to be calculated, the input vector $x \in \mathcal{R}^m$ and the input matrix $A \in \mathcal{R}^{3 \times m}$. Assume that we have $n=2$ worker nodes and matrix $A$ is divided into two submatrices of $A_1 \in \mathcal{R}^{2 \times m}$ and $A_2 \in \mathcal{R}^{1 \times m}$. Matrix multiplication can be done in parallel by allowing the first worker to compute $y_1=A_{1}x$ and the second worker $y_2=A_{2}x$. Here, the job size is $3$ CUs, the first worker's task size is $2$ CUs, and the second worker's task size is $1$ CU. In this study, we assume that the size of the job, \(J\), is proportional to the number of workers \(n\), that is, \(J = \mathcal{O}(n)\). For our theoretical analysis, we define the job size as \(J = nl\). We refer to \(l\) as the job scale factor.
\\[1ex]
\ul{Erasure Coding Model:}
To create redundant tasks, the job of size $J$ CUs is partitioned into $k$ tasks of equal size $s=\frac{J}{k}$ CUs. Next, 
MDS (maximum distance separable) erasure coding is used to create $n-k$ redundant tasks with code rate $R=\frac{k}{n}$. The redundant tasks are generated by erasure coding. Since MDS coding is used, the job completes when any $k$ of $n$ tasks are executed. We restrict our analysis to MDS coding as its features are convenient enough to show the analysis between the batch size and the redundancy level.
\\[1ex]
\ul{System Architecture and Task Batching Schemes:} We consider a distributed computing system, illustrated in Fig.~\ref{fig:System_Model}, consisting of a single master node and $n$ worker nodes, typical in systems such as Apache Mesos~\cite{hindman2011mesos}, Apache Spark~\cite{zaharia2010spark}, and MapReduce~\cite{dean2008mapreduce}. Following the MDS erasure coding model, a job of size $J$ is divided into $n$ tasks, with task size determined by the redundancy level, $s = \frac{J}{k}$ CUs. Each task is further partitioned into batches of size $b$ CUs, resulting in $G = \frac{s}{b}$ batch generations. A batch generation $G_i$ is considered complete once any $k$ out of $n$ workers finish their assigned batch tasks. The workers then proceed to $G_{i+1}$, until the job completes. Fig.~\ref{fig:System_Model} provides an example with $J=6$ CUs, maximum redundancy (replication), and batch size $b=2$ CUs.
\begin{figure}[t]
    \centering
    \includegraphics[scale=0.12]{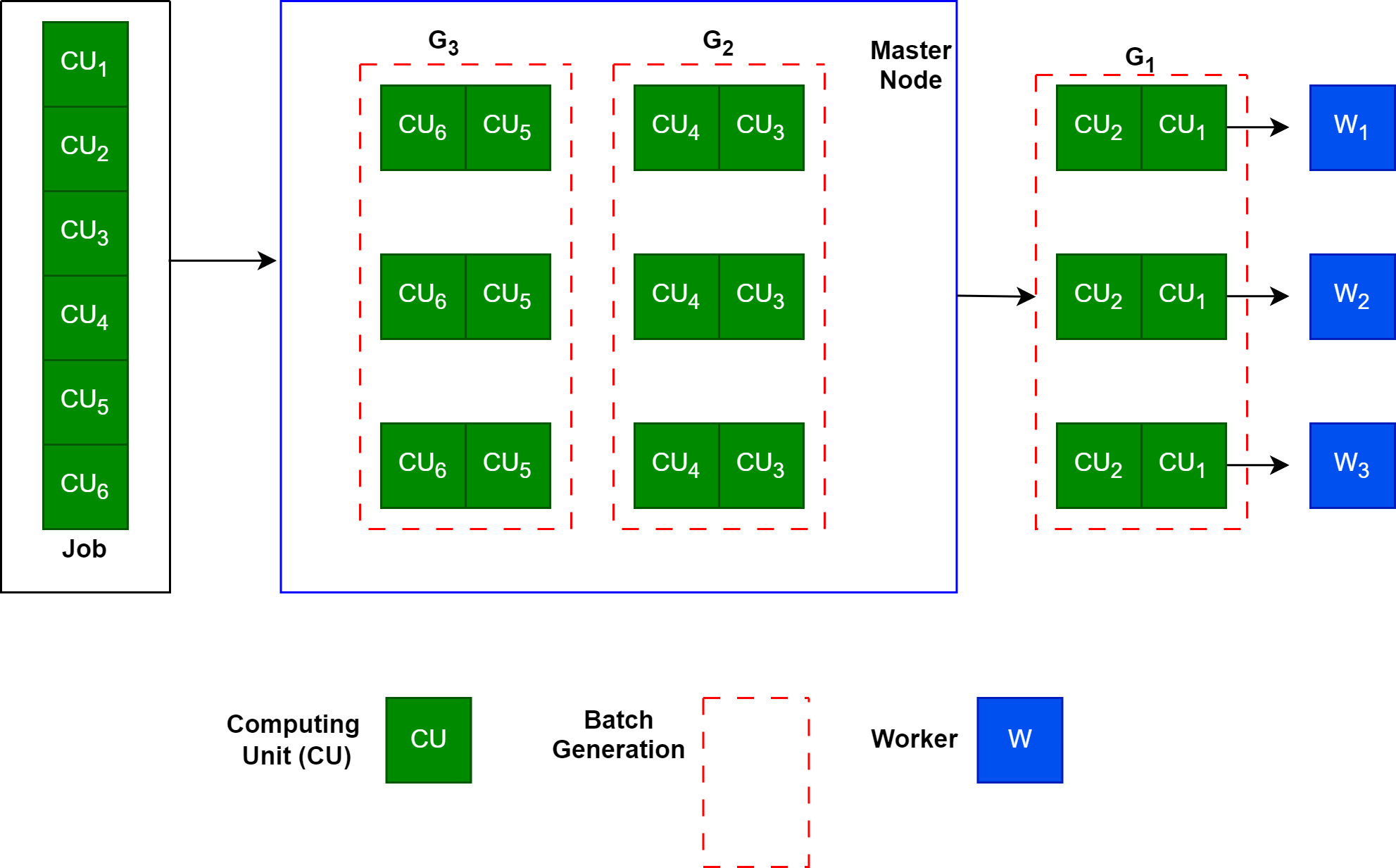}
    \caption{Distributed Computing System. Here, job size $J=6$ CUs is replicated, and $3$ batch generations are formed for selecting $b=2$ CUs. The ongoing batch generation $G_1$ gets completed when any $1$ worker finishes the batch task. The job will be completed when the $G_3$ batch generation is executed. }
    \label{fig:System_Model}
\end{figure}
\\[1ex]
\ul{Service Time PDF and Scaling:} The service time at each worker node is random and needs to scale with the task size. 
We chose the two most common service time PDFs used frequently in the literature. The first is the shifted exponential distribution \cite{joshi2014delay,aktas2017simplex,wang2021batch}, and the other is the bi-modal distribution \cite{behrouzi2019redundancy}. To model how the consecutive CUs will be executed, we will use the additive scaling model \cite{kiani2018exploitation,peng2021diversity}. Here, the execution of each CU is i.i.d. So the batch task completion time by a worker $i$ is $Y_i^b=X_1+X_2+...+X_b$, where $X_1, X _2...X_b$ are independent. 

\section{Performance Metric and Problem Statement}\label{sec:evaluation_metric}
\subsection{Performance Metric: Job Completion Time}

As for the performance evaluation metric, we choose the expected job completion time. Based on the system model (Sec.~\ref{sec:system_Model}), we calculate the expected job completion time as follows. For any arbitrary batch generation \(G_g\), the batch completion time is defined as the \(k\)-th fastest worker batch task execution time. For worker \(i\) in batch generation \(G_g\), we denote the batch task completion time as \(\leftindex_g Y^b_i\). The batch generation completion time is then equal to the \(k\)-th fastest worker's batch task completion time, denoted as \(\leftindex_g Y^b_{k:n}\). Mathematically, \(\leftindex_g Y^b_{k:n}\) represents the \(k\)-th order statistic of the \(n\) samples taken from the distribution of \(Y^b\). Finally, the job completion time, \(\leftindex^G Y^b_{k:n}\), is the sum of the \(G\) batch generations' completion times: $\leftindex^G Y^b_{k:n} = \sum_{g=1}^G \leftindex_g Y^b_{k:n}$. Since every batch generation is executed independently, the expected job completion time can be expressed as:

\begin{equation}
\begin{aligned}
    \label{equ:ejct}
    \mathbb{E}\left[\leftindex^G Y^b_{k:n}\right] = G \mathbb{E}\left[Y^b_{k:n}\right]=\frac{l}{Rb}\mathbb{E}\left[Y^b_{k:n}\right]
\end{aligned}
\end{equation}
The above expression expresses the tradeoff between the expected job completion time and batch size. For a fixed redundancy level (code rate $R$), if the batch
size $b$ decreases, the expected batch completion $\mathbb{E}\left[Y^b_{k:n}\right]$ decreases, but now we have more batch generations to complete, and vice versa. In subsequent sections, we will demonstrate how the optimal choice of batch size is influenced by the redundancy level and the service time probability density function (PDF).

\subsection{Problem Statement and Summary of Results}
\begin{table*}
\centering
\renewcommand{\arraystretch}{1.3}
\setlength{\tabcolsep}{4pt} 
\caption{Optimal Batch Size for Fixed Redundancy}
\begin{tabular}{|c|c|}
\hline
\textbf{Service Time PDF} & \textbf{Optimal Batch Size} \\ \hline
Shifted Exponential & \makecell{Either minimum or maximum batch size.\\ The decision depends on given code rate $R$ and job size.} \\ \hline
Bi-Modal & \makecell{Sequentially shifts from minimum size to maximum size. The decision \\depends on the straggling probability $\epsilon$ and given code rate $R$.} \\ \hline
\end{tabular}
\label{tab:red_vs_batch}
\end{table*}



\begin{table*}[htbp]
\centering
\renewcommand{\arraystretch}{1.3} 
\setlength{\tabcolsep}{6pt} 
\caption{Optimal Strategy}
\begin{tabular}{|c|c|c|c|}
\hline
\textbf{Service Time PDF} & \textbf{Low Straggling} & \textbf{Medium Straggling} & \textbf{High Straggling} \\ \hline
\textbf{\makecell{Shifted Exponential}} 
& \makecell{Splitting with maximum batch size} 
& \makecell{Low $l$: Coding with minimum batch size \\ High $l$: Splitting with maximum batch size} 
& \makecell{Replication with minimum batch size} \\ \hline
\textbf{\makecell{Bi-Modal}} 
& \makecell{Splitting with maximum batch size} 
& \makecell{Coding with minimum batch size} 
& \makecell{Splitting with maximum batch size} \\ \hline
\end{tabular}
\label{tab:summ_table_b_R}
\end{table*}

Our study on determining the optimal batch size $b^*$ has two main objectives. First, we aim to investigate how the optimal batch size $b^*$ is related to the redundancy level in a fixed number of workers (\(n\)). For a given redundancy level, defined by the code rate \(R = \frac{k}{n}\), the task size is determined as \(s = \frac{J}{k}\). The choice of batch size \(b\) is constrained such that $b \in \{x \in \mathcal{N}: s \mod x =0,1 \leq x \leq s\}$
Here, we assume that each batch generation consists of equal-sized batches, each containing \(b\) computation units (CUs). The goal is to determine the batch size \(b\) that minimizes the expected job completion time \(\mathbb{E}[Y^b_{k:n}]\). This problem can be formulated as the following optimization: $b^* = \argmin_{b} \mathbb{E}[Y^b_{k:n}]$
subject to a fixed  \(R = \frac{k}{n}\). Table \ref{tab:red_vs_batch} summarizes the results of our analysis, showing the relationship between redundancy levels and optimal batch size for two different service time PDFs. The findings indicate that the optimal batch size decision is influenced by the underlying service time distribution, highlighting the importance of accounting for the service time characteristics when optimizing batch size.
Next, we optimize the expected job completion time with respect to both the redundancy level (code rate \(R\)) and the batch size \(b\). For a fixed number of workers (\(n\)), we can vary the redundancy level across three cases: splitting (\(R = \frac{1}{n}\)), coding (\(\frac{1}{n} < R < 1\)), and replication (\(R = 1\)). For each choice of redundancy, the task size \(s\) changes, and so does the choice of batch size \(b\). \emph{Our goal is to simultaneously optimize the code rate \(R^*\) and batch size \(b^*\) to minimize the expected job completion time:} $[b^*, R^*] = \argmin_{b, R} \mathbb{E}[Y^b_{k:n}]$.
Table \ref{tab:summ_table_b_R} summarizes the optimal strategies for the two service time PDFs. The results indicate that the optimal strategy for the same straggling scenario varies depending on the underlying service time distribution.

\section{Shifted Exponential Service Time}\label{sec:shifted_exponential}

Under the Shifted Exponential service time model, each CU execution follows the shifted exponential distribution, i.e., $X \sim S-Exp(\Updelta,W)$, where $\Updelta$ is the minimum execution time and $W$ is the straggling parameter. The straggling effect increases with the increase of $W$. The expected value of a single CU execution time is $\expect{X}=\Updelta+W$.

\begin{figure*}[hbt]
    \centering
    \begin{subfigure}[b]{0.3\textwidth}
        \centering
        \includegraphics[width=\textwidth]{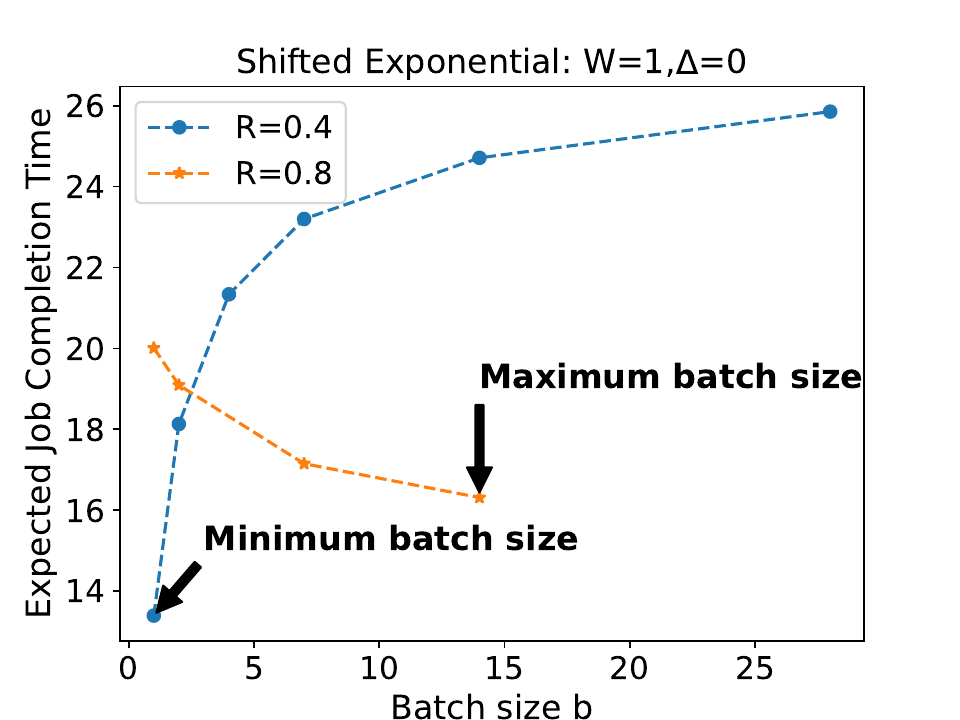}
        \caption{$J=112$ CUs, $n=10$}
        \label{fig:min_max_exp_4_8}
    \end{subfigure}
    \begin{subfigure}[b]{0.3\textwidth}
        \centering
        \includegraphics[width=\textwidth]{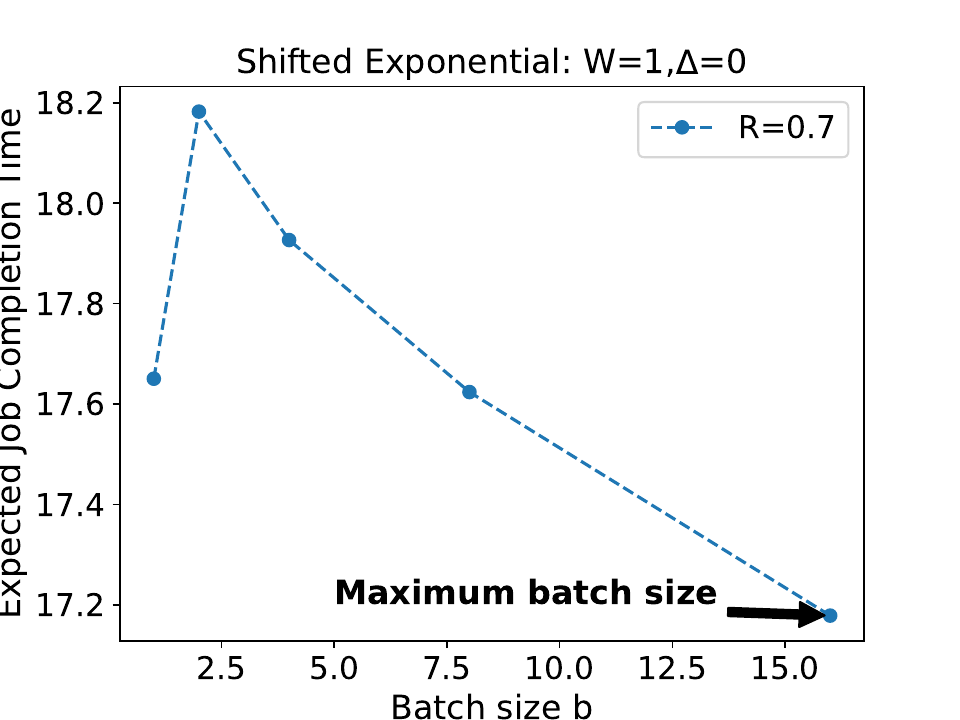}
        \caption{$J=112$ CUs, $n=10$}
        \label{fig:max_exp_7}
    \end{subfigure}
    \begin{subfigure}[b]{0.3\textwidth}
        \centering
        \includegraphics[width=\textwidth]{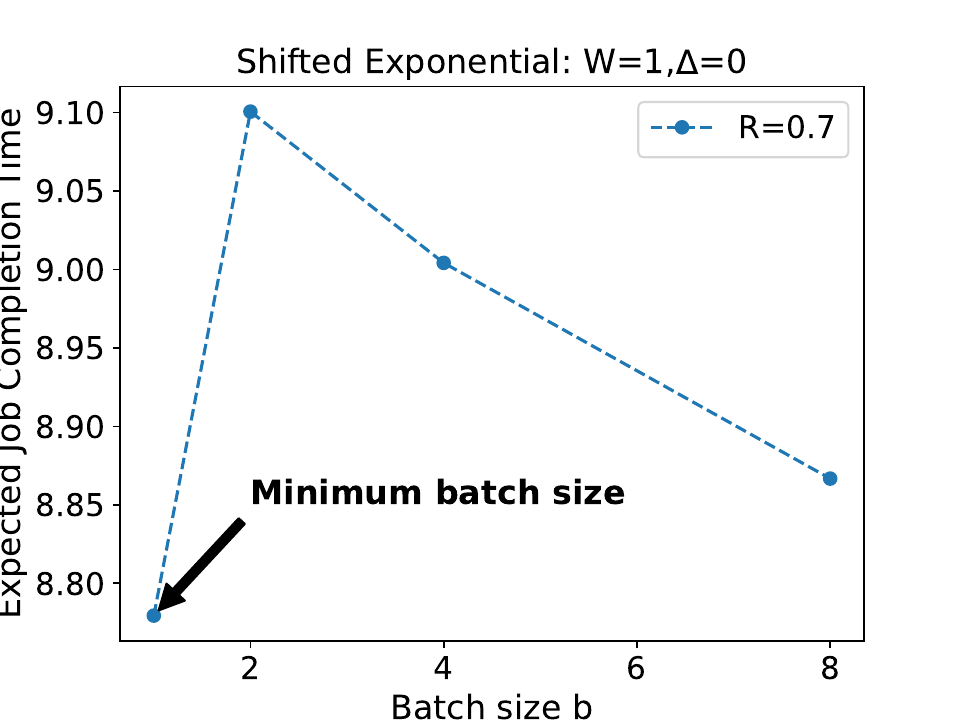}
        \caption{$J=56$ CUs, $n=10$}
        \label{fig:min_exp_7}
    \end{subfigure}
    \caption{Optimizing batch size for different redundancy levels. (a) The optimum batch size is the maximum for redundancy level $R=0.8$ and is the minimum for redundancy level $R=0.4$.
    For (b) and (c), we have the same redundancy level, $R=0.7$, but different job sizes. (b) is optimum with maximum batch size. (c) is optimum with minimum batch size.}
    \label{fig:exo_b_vs_R}
\end{figure*}

\subsection{Optimum Batch Size and Redundancy }
We start the analysis by presenting the optimum batch size for the two extreme cases, namely for splitting ($R=\frac{1}{n}$) and replication ($R=1$). Note that these results apply to any general service time PDF. 
\begin{theorem}\label{th:rep}
For replication, job completion time does not increase as the batch size decreases. The optimum batch size is a single computing unit, $b^*=1$ CU.
\end{theorem}
\begin{proof}
Let's assume each worker gets a task of size $s$  for selecting redundancy parameter $k=1$ ($R=\frac{1}{n}$). If the batch size is $b=s$, then there is single batch generation $G=1$, and the job completion time will be $\leftindex^1 Y^s_{1:n}=\min{\{\leftindex_1Y_1^s,\leftindex_1Y_2^s\dots\leftindex_1Y_n^s\}}$. The following inequalities show that the system finishes the job faster if we create more batch generations, $G>1$, with lower batch size $b<s$. 
\begin{equation}\label{equ:proof_max_redundancy}
\begin{aligned}
& \min{\{\leftindex_1Y_1^s,\leftindex_1Y_2^s,\dots,\leftindex_1Y_n^s\}}\\ &= \min{\{(\leftindex_1Y_1^{s-1}+\leftindex_2Y_1^{1}),(\leftindex_1Y_2^{s-1}+\leftindex_2Y_2^{1}).....(\leftindex_1Y_n^{s-1}+\leftindex_2Y_n^{1})\}} \\
    & \geq \min{\{\leftindex_1Y_1^{s-1},\leftindex_1Y_2^{s-1}},\dots,\leftindex_1Y_n^{s-1}\} + \min{\{{\leftindex_2Y_1^{1},\leftindex_2Y_2^{1}},\dots,\leftindex_2Y_n^{1}\}\}}\\
    & \shortvdotswithin{ = }\notag \\[-3.5ex]
    &\geq \sum_{i=1}^{G} \min{\{\leftindex_gY_1^{1},\leftindex_gY_2^{1}},\dots,\leftindex_gY_n^{1}\}
\end{aligned}
\end{equation}
In the first equality, we write the task execution time $\leftindex_1Y_j^s$ of size $s$ as the subsequent executions of size $s-1$ and unit task.
The first inequality follows from the fact that $\min{\{A+B,C+D\}}\geq\min{\{A,C\}} + \min{\{B,D\}}$.
Thus, for replication, the job completion time decreases as we split the task into more batches, and the lowest job completion time will be with the lowest possible batch size, which is $b=1$ CU.
\end{proof}
\begin{theorem}\label{th:splitting}
    For splitting, job completion time is non-decreasing as the batch size decreases. The optimum batch size is equal to the task size, i.e., $b^*=s$ 
\end{theorem}
\begin{proof}
Let's assume each worker gets a task of size $s$ for selecting redundancy parameter $k=n$ ($R=1$). If there is single batch generation $G=1$ and batch size is $b=s$, the job completion time will be $\leftindex^1 Y^s_{1:n}=\max{\{\leftindex_1Y_1^s,\leftindex_1Y_2^s,\dots,\leftindex_1Y_n^s\}}$. We claim the system can finish the job more slowly if we create more batch generations, $G>1$, with a lower batch size $b<s$. This can be shown as follows.
\begin{equation}\label{equ:proof_min_redundancy}
\begin{aligned}
\max{\{\leftindex_1Y_1^s,\leftindex_1Y_2^s,\dots,\leftindex_1Y_n^s\}}
\leq \sum_{i=1}^{G} \max{\{\leftindex_gY_1^{1},\leftindex_gY_2^{1}},\dots,\leftindex_gY_n^{1}\}
\end{aligned}
\end{equation}
which is based on the following property: $\max{\{A+B,C+D\}}\leq\max{\{A,C\}} + \max{\{B,D\}}$. Thus, the optimum batch size is the highest possible, which is $b=s$ CUs.
\end{proof}
We next consider the general code rate, $\frac{1}{n} < R < 1$. We will start our analysis by presenting the expected job completion time. The exact expression is unsuitable for theoretical analysis; we derive the expression for a large value of $n$ (Appendix in \cite{optimal_batch_size_coded_computing}). It is worth noting that although our theoretical results are based on a large scale size of $n$, they are useful even for a small value of worker number and job size, as demonstrated in the simulation results.

\newpage
\noindent
\ul{Expected Job Completion Time:}
The expected job completion time for $X \sim S-Exp(\Delta,W)$ of job size $J=nl$ CUs, as $n\to \infty$, is given as
\begin{equation}\label{equ:expo_shifted_ejct}
\frac{1}{W}\expect{\leftindex^GY^b_{k:n}}= l\Big(\frac{\Updelta}{RW}+\frac{m}{Rb}\Big) \quad  n \to \infty
\end{equation}
where $m$ is the solution of $P(b,m)=R$ and $P(b,m)$ is regularized lower incomplete gamma function. By definition, $P(b,m)=\frac{1}{(b-1)!}\int_0^m t^{b-1} e^{-t}dt=1 - e^{-m}\sum\limits_{k=0}^{b-1}\frac{m^k}{k!}
$. \\
The first part of the equation is due to the minimum completion time of each CU, and the second part is due to the randomness of the task execution. In Theorem \ref{thm:b1_bmax}, we minimize \eqref{equ:expo_shifted_ejct} with respect to $b$ to get the optimum batch size for an arbitrary fixed redundancy level (code rate $R$).

\begin{theorem}\label{thm:b1_bmax}

For $X \sim S-Exp(\Updelta,W)$ and for an arbitrary fixed code rate $R$, the optimum batch size is either $b^*=1$ or $b^*=s$ (task size).

\end{theorem}

\begin{proof}
Equ.\eqref{equ:expo_shifted_ejct} as a function of variables \( b \) and \( R \), denoted as \( f(b,R) = \frac{m}{R b} \). We will show that $f(b,R)$ is either a monotone or a 
unimodal function with respect to $b$, which proves the theorem. 
Assuming \( b \) is continuous, the derivative of the function \( f(b,R) \) with respect to \( b \) can be expressed as follows:
    \begin{equation} \label{equ:diff_f}
        f'(b,R)=\frac{\partial(b,R)}{\partial b}=\frac{1}{Rb}\Big(\frac{dm}{db}-\frac{m}{b}\Big)
    \end{equation}
Based on the definition (defined at Equ.\eqref{equ:expo_shifted_ejct}), the batch size $b$ is analogous to the number of Poisson events $k$, $m$ is analogous to the Poisson rate parameter $\lambda$, and CDF of the $k-1$ Poisson events is equal to the $1-R$. So, using the Normal approximation to the Poisson distribution, batch size $b$ can be written as follows:
\begin{equation}\label{equ:b_approx_m}
    b \approx Z\sqrt{m} + m+ 1
\end{equation}
where $\Phi(Z)=1-R$ and $\Phi(\cdot)$ is the area of the left tail of the standard normal CDF.
From \eqref{equ:b_approx_m}, we can find the expression of $\frac{dm}{db}=\frac{2\sqrt{m}}{Z+2\sqrt{m}}$. Thus Equ.\eqref{equ:diff_f} can be rewritten as follows
\begin{equation}\label{equ:Final_equ_diff_f}
    f'(b,R) = \frac{1}{Rb} \Big(\frac{2\sqrt{m}}{Z+2\sqrt{m}}-\frac{m}{b}\Big)
\end{equation}
Inspection shows that there are code rate $R_1$ and $R_2$ $(R_1<R_2)$, for which
$f'(b,R)$ is always positive for $R<R_1$, and is always negative for $R>R_2$. So, $f(b,R)$ is either an increasing or decreasing function for the corresponding code rate region. For $R_1 < R < R_2$, there is batch size $b'>1$ for which $f'(b,R)$ is positive for $b<b'$ and negative for $b>b'$. This implies that $f(b,R)$ is a unimodal 
function for $R_1<R<R_2$.
\end{proof}
\noindent
\ul{Simulation Result:}
We simulated the expected job completion time for two different job sizes, $J = 112$ CUs and $J = 56$ CUs, and three different code rates, $R = 0.4$, $0.7$, and $0.8$ (Fig.~\ref{fig:exo_b_vs_R}). For Fig. \ref{fig:min_max_exp_4_8}, we have job size $J=112$ CUs and $n=10$ workers. We found that the optimum batch size is the minimum for $R=0.4$ and maximum for $R=0.8$. For 
Fig.~\ref{fig:max_exp_7} and Fig.~\ref{fig:min_exp_7}, we have the same redundancy levels and number of workers, $R=0.7$ and $n=10$, but different job sizes. 
Fig.~\ref{fig:max_exp_7} has job size $J=112$ CUs, and 
Fig.~\ref{fig:min_exp_7} has job size $56$ CUs. The figures are consistent with Theorem
\ref{thm:b1_bmax}, showing the optimum batch size is either at the minimum or the maximum despite a moderate scale factor $n = 10$. 


The optimal batch size for larger jobs is the maximum (Fig. \ref{fig:max_exp_7}). The optimum batch size for smaller jobs is the minimum (Fig. \ref{fig:min_exp_7}). \\
From the above discussion, we can conclude that there is a coding rate $R'$ for which the optimal batch size is always the maximum (task size $s$) when $R>R'$. The theoretical value of $R'$, $R' \approx 0.72$, is obtained in the following theorem. 

\begin{theorem}
    For $X \sim S-Exp(W,\Updelta)$, $b^*=s$ CU is the optimal batch size for $R>R'$, where $R'>0$ is the unique positive solution of the following system of equations: $R'=1-e^{-m_1}$ and $e^{m_1}=1+2m_1$.
\end{theorem}

    



    

\subsection{Optimize Expected Job Completion Time}
\begin{figure*}[hbt!]
    \centering
    \begin{subfigure}[b]{0.24\textwidth}
        \centering
        \includegraphics[width=\textwidth]{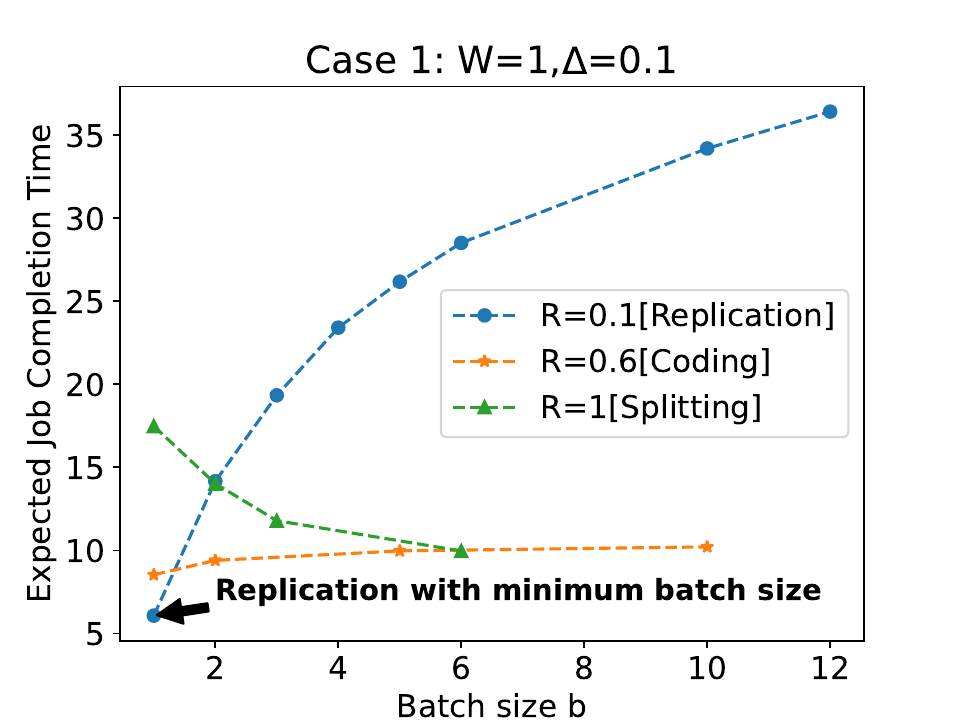}
        \caption{$J=60$ CUs, $n=10$}
        \label{fig:W_1_delta_0}
    \end{subfigure}
    \begin{subfigure}[b]{0.24\textwidth}
        \centering
        \includegraphics[width=\textwidth]{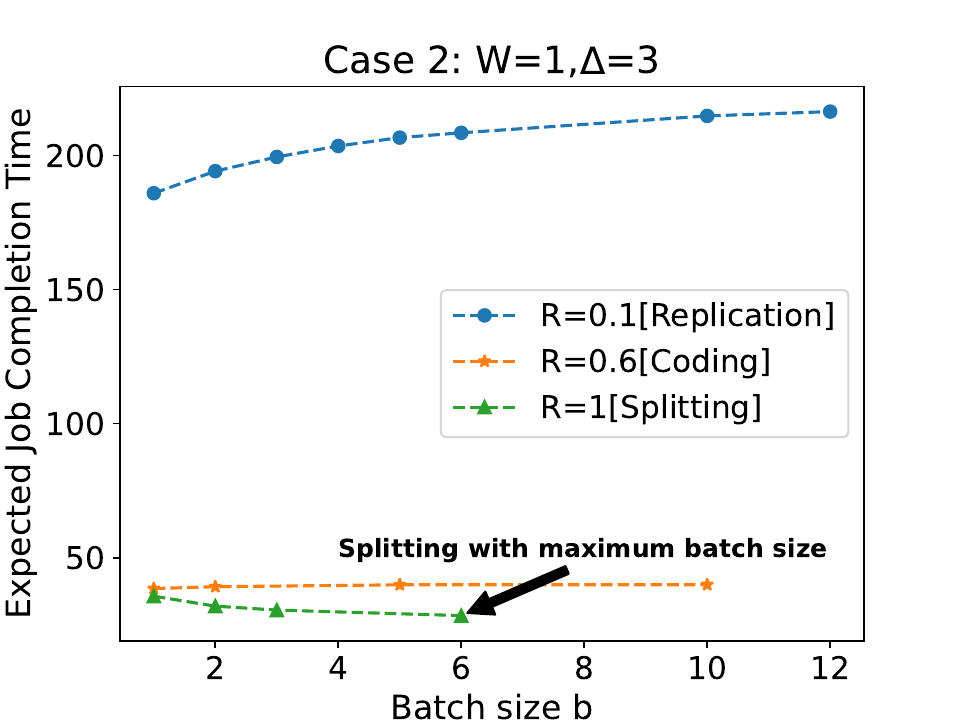}
        \caption{$J=60$ CUs, $n=10$}
        \label{fig:W_1_delta_3}
    \end{subfigure}
    \begin{subfigure}[b]{0.24\textwidth}
        \centering
        \includegraphics[width=\textwidth]{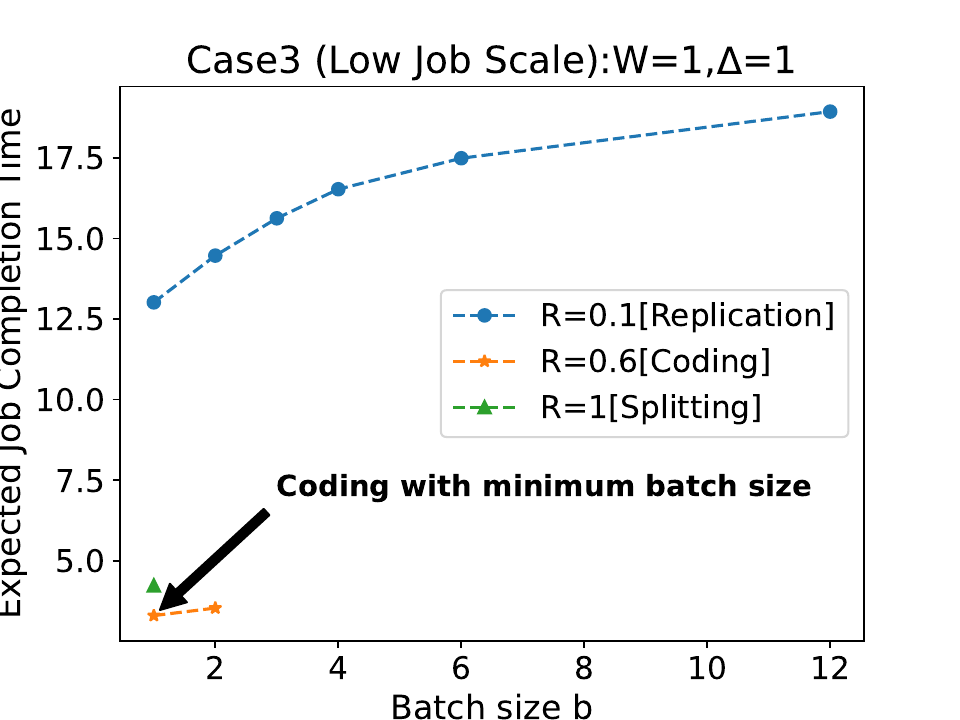}
        \caption{$J=12$ CUs, $n=12$}
        \label{fig:W_1_delta_1_low_job}
    \end{subfigure}
    \begin{subfigure}[b]{0.24\textwidth}
        \centering
        \includegraphics[width=\textwidth]{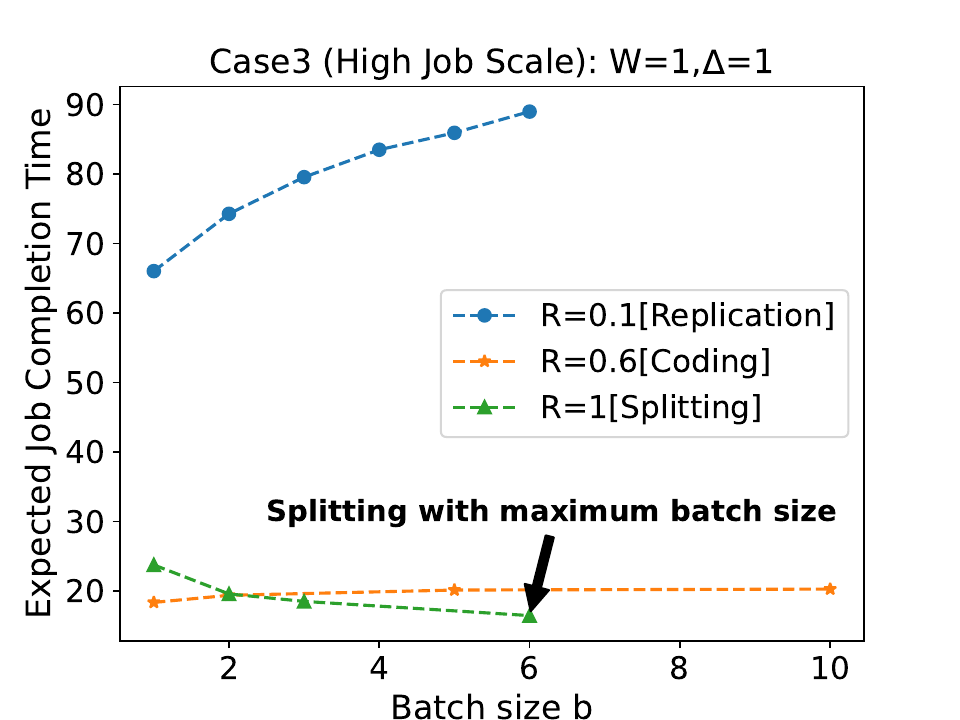}
        \caption{$J=60$ CUs, $n=10$}
        \label{fig:W_1_delta_1_high_job}
    \end{subfigure}
 \caption{Optimizing the expected job completion time: (a) and (b) have the same job size and number of workers. The optimal strategy is replication with minimum batch size for $W=1,\Updelta=0.1$ and splitting with minimum batch size when $W=1,\Updelta=10$. Cases (c) and (d) have $W=\Updelta=1$ and differ by the job scale factor. Coding with the minimum batch size is optimal for the low job scale factor, and splitting with the maximum batch size is optimum for the high job scale factor.}
    \label{fig:exp_optimize_ejct}
\end{figure*}

Next, we identify the optimal combination of code rate and batch size, denoted as $(R^*, b^*)$, which minimizes the expected job completion time across various straggling scenarios. We consider three cases depending on the ratio $\frac{W}{\Updelta}$, which characterizes the randomness in the system. The next three separate theorems give the results.
\begin{theorem}\label{thm:case1}
      For $X \sim S-Exp(\Updelta,W)$ and $\frac{W}{\Updelta} \gg 1$, the expected job completion time is minimized by replication with batch size $b=1$ for a large value of $n$.  
\end{theorem}
\begin{proof}
For batch size $b=1$ and for $\frac{W}{\Updelta} \gg 1$ expected job completion time , Equ.\ref{equ:expo_shifted_ejct} can be written as follows
\[\lim_{\frac{\Updelta}{W} \to 0} \expect{\leftindex^GY_{k:n}^{b=1}} = \frac{lWP^{-1}(R,1)}{R} =-\frac{lWlog(1-R)}{R} \]
which is a strictly increasing function with code rate $R$. So, for the batch size $b=1$, the optimal code rate will be the lowest possible code rate, i.e., $R=\frac{1}{n}$(replication). For batch size $b=1$ with replication $k=1$, the expected job completion time can be derived as follows
\begin{equation}\label{equ:case_1_r_min}
        \lim_{\frac{\Updelta}{W} \to 0} \expect{\leftindex^GY_{1:n}^{b=1}} =G\frac{W}{n}=lW\\
\end{equation}
We next calculate the expected job completion time for the maximum batch size, where $b_{max}=s=\frac{l}{R}$.


\[\lim_{\frac{\Updelta}{W} \to 0} \expect{\leftindex^GY_{k:n}^{b=\frac{l}{R}}} = WP^{-1}(R,\frac{l}{R}) =mW \].

For a high job scale factor, $l \gg 1$, $P^{-1}(R,\frac{l}{R})=m$ can be considered a moderately high value to use the Normal approximation to the Poisson distribution. 
\begin{equation}
    l \approx R (Z\sqrt{m} + m)
\end{equation}
where $\upphi(Z)=1-R$ and $\upphi(\cdot)$ is the area of the left tail of the standard normal CDF.
$Z$ is close to $0$ for code rate $R \approx \frac{1}{2}$ and is negative for $R=1$(splitting). So, for both cases, $l < m$. which implies that replication ($R=\frac{1}{n}$) with unit batch size $b=1$ is optimal.
It can be shown that numerically, $l < m$ also holds for a low job scale factor, $l \approx 1$.
\end{proof}

\begin{theorem}\label{thm:case2}
    For large $n$ and $\frac{W}{\Updelta} \rightarrow 0$, splitting with the maximum batch size is the optimum strategy.  
\end{theorem}

%



The above result is expected if we analyze the behavior of the system. Informally, when $\frac{W}{\Updelta} \to 0$, little randomness is present in the execution of the unit task and among the workers. All workers are "equally" fast or slow, determined by the value $\Updelta$. In the splitting regime with the maximum batch size $b_{max}=s$, all $n$ workers will complete their single batch generation more or less simultaneously with $s\Updelta$.

\begin{theorem}\label{thm:case3}
For large $n$ and $\frac{W}{\Updelta} \rightarrow 1$,
coding with unit batch size completes the job faster than splitting for low job scale factor \(l\). For a high job scale factor \(l\), the relationship is inverse. 
\end{theorem}


%
%
%
%
\noindent
\ul{Simulation Results:}
We simulated the expected job completion time for three different $\frac{W}{\Updelta}$ ratios and two different job scale factors (Fig. \ref{fig:exp_optimize_ejct}). Fig. \ref{fig:W_1_delta_0} $(\frac{W}{\Updelta}=10, W=1, \Updelta=0.1)$ is optimized by replication with minimum batch size, Fig. \ref{fig:W_1_delta_3} $(\frac{W}{\Updelta}=\frac{1}{3}, W=1, \Updelta=3)$ is optimized by splitting with maximum batch size. Both Fig. \ref{fig:W_1_delta_1_low_job} and Fig. \ref{fig:W_1_delta_1_high_job} have $W=\Updelta=1$, but have different job scale factor $l$. Fig. \ref{fig:W_1_delta_1_low_job} is optimized by coding with minimum batch size (job scale factor $l=1$), and Fig. \ref{fig:W_1_delta_1_high_job} is optimized by splitting with maximum batch size (job scale factor $l=6$).


\section*{Acknowledgment}
This work was supported in part by NSF CCF-2327509. The authors thank M.~Aktas and P.~Peng for valuable discussions.


\newpage
\bibliographystyle{IEEEtran}
\bibliography{bibliofile}

\end{document}